\documentclass[12pt,a4paper]{article}

\topmargin by -\headsep \textheight 8.9in \oddsidemargin 0pt
\evensidemargin \oddsidemargin \marginparwidth 0.5in \textwidth
6.5in

\usepackage{graphicx}
\usepackage{amsmath,amssymb,amsfonts,amsthm}

\def\be{\begin{equation}}
\def\ee{\end{equation}}
\def\bea{\begin{eqnarray}}
\def\eea{\end{eqnarray}}
\def\bma{\begin{mathletters}}
\def\ema{\end{mathletters}}

\def\0{\overline{0}}

\def\q0{\underline{0}}

\def\H{{\cal H}}

\newcommand{\C}{{\mathbb{C}}}
\def\deg{\mbox{deg}}

\def\W{{\cal W}}

\def\R{{\mathbb R}}
\def\N{{\mathbb N}}
\def\C{{\mathbb C}}

\def\tr{\mbox{tr}}
\def\one{\leavevmode\hbox{\small1\normalsize\kern-.33em1}}

\def\bra#1{\langle#1|} \def\ket#1{|#1\rangle}
\def\braket#1#2{\langle#1|#2\rangle}

\def\proj#1{\ket{#1}\!\bra{#1}}

\newcommand\openone{\leavevmode\hbox{\small1\kern-3.8pt\normalsize1}}

\newtheorem{theo}{Theorem}

\newtheorem{lemma}[theo]{Lemma}
\newtheorem{prop}[theo]{Proposition}

\def\id{{\mathbb I}}
\def\tr{\mbox{tr}}

\begin{document}

\title{A paradox in bosonic energy computations via semidefinite programming relaxations}
\author{M. Navascu\'es$^1$, A. Garc\'ia-S\'aez$^2$, A. Ac\'in$^3$,  S. Pironio$^4$, and M. B. Plenio$^5$\\
\small $^1$School of Physics, University of Bristol, BS8 1TL, Bristol, U.K.\\
\small $^2$Dept. d'Estructura i Constituents de la Mat\`{e}ria, Universitat de Barcelona, 08028 Barcelona, Spain\\
\small $^3$ICFO-Institut de Ci\`{e}ncies Fot\`{o}niques, E-08860 Castelldefels (Barcelona), Spain\\
\small $^4$Laboratoire d'Information Quantique, Universit\'e Libre de Bruxelles, 1050 Bruxelles, Belgium\\
\small $^5$Universit\"{a}t Ulm, Institut f\"{u}r Theoretische Physik, 89069 Ulm, Germany.}

\date{}
\maketitle

\begin{abstract}
We show that the recent hierarchy of semidefinite programming relaxations based on non-commutative polynomial optimization and reduced density matrix variational methods exhibits an interesting paradox when applied to the bosonic case: even though it can be rigorously proven that the hierarchy collapses after the first step, numerical implementations of higher order steps generate a sequence of improving lower bounds that converges to the optimal solution. We analyze this effect and compare it with similar behavior observed in implementations of semidefinite programming relaxations for commutative polynomial minimization. We conclude that the method converges due to the rounding errors occurring during the execution of the numerical program, and show that convergence is lost as soon as computer precision is incremented. We support this conclusion by proving that for any element $p$ of a Weyl algebra which is non-negative in the Schr\"{o}dinger representation there exists another element $\tilde{p}$ arbitrarily close to $p$ that admits a sum of squares decomposition. 
\end{abstract}

\section{Introduction}
Computing the energy spectrum of a finite set of indistinguishable particles subject to a given potential is a standard problem appearing in many branches of physics, e.g., in quantum chemistry, atomic physics, or condensed matter physics.
Although traditionally the main approaches to this problem have been variational \cite{szabo}, in the last decade it has been attacked with success by means of semidefinite programming (SDP) formulations \cite{nakata, mazziotti, nakata2} of the constraints on second-order reduced density matrices proposed in \cite{coleman,garrod,erdahl} (the so-called 2-RDM method).
These SDP methods can be viewed as particular instances of a more general \emph{non-commutative polynomial optimization} approach \cite{POI,chapter}, which extends to the non-commutative setting the method developed by Lasserre \cite{moment} and Parrilo \cite{parrilo} for scalar polynomial optimization. Roughly speaking, a non-commutative optimization problem consists in finding the minimal eigenvalue of a hermitian polynomial of non-commutative operators. To solve such problems, one can define a hierarchy of SDP relaxations, each of which corresponds to finding a sum of squares decomposition of the polynomial to be minimized which provides a lower bound $p^{k}$ on the optimal solution $p^\star$ of the original problem, with $p^1\leq p^2\leq...\leq p^\star$.
 This approach reduces to the 2-RDM method when applied to fermionic systems, but more generally is also highly successful, e.g., to characterize the set of quantum correlations in quantum information science \cite{quantum}. Recently, modifications of this algorithm exploiting translational invariance have been proposed independently by H\"{u}bener and Barthel \cite{hubener} and Baumgratz and Plenio \cite{tillmann} as an alternative to variational techniques in condensed matter physics.
In \cite{maz_bos,hubener,tillmann}, it was suggested to apply such SDP methods to compute the ground-state energy of bosonic systems, i.e., to find the minimal eigenvalue of Weyl polynomials. 

In this article, we point out that any computer implementation of the SDP hierarchies \cite{POI,chapter} to bosonic systems will exhibit the non-commutative analog of an effect already observed in similar algorithms for commutative polynomial minimization \cite{henrion}, \cite{waki}. On one hand, it can be proven that any relaxation beyond the first one will not provide better lower bounds on $p^\star$, i.e., $p^k=p^1$ for all $k\geq 1$. On the other hand, numerically it is observed that the bounds $\hat{p}^1, \hat{p}^2,...$ output by the computer form an increasing sequence, with $\lim_{k\to\infty} \hat{p}^k=p^\star$.

We will show that the resolution of this ``mathematical paradox'' follows the same lines as the commutative one. Namely, even though there exist positive Weyl polynomials $p$ that do not admit a sum of squares decomposition, for any such polynomial there exists an arbitrarily small perturbation $p\to\tilde{p}= p+\epsilon g$ such that $\tilde{p}$ can be decomposed as a sum of squares of Weyl polynomials. The rounding errors introduced by the computer while executing the algorithm correspond to such a perturbation, and so numerical implementations of the SDP method converges to the correct answer of the problem. This is, therefore, an example of a numerical method that converges, not in spite of rounding errors, but \emph{because} of them.

This article is structured as follows. In Section \ref{weyl}, we provide the basic definitions and facts about Weyl algebras that are used in the remaining of the text. In Section \ref{method} we will describe the SDP method and illustrate the paradox with a numerical example. The resolution of the paradox will be given in Section \ref{positivstellensatz}, where we will prove that any positive Weyl polynomial can be perturbed to a sum of squares of polynomials.
Finally, in Section \ref{conclusion} we will present our conclusions.

\section{Definitions and basic results on Weyl algebras}
\label{weyl}
\subsection{Weyl algebras and Weyl polynomials}
A \emph{Weyl algebra} ${\cal W}_n$ is a $*$-algebra with $2n$ generators $\bar a=(a_1,a_2,...,a_n)$ and $\bar{a}^*=(a^*_1,a^*_2,...,a^*_n)$ satisfying the canonical commutation relations (CCRs):
\be
[a_i,a_j]=[a_i^*,a_j^*]=0,\,[a_i,a^*_j]=\delta_{ij}, \mbox{ for all } i,j\in\{1,...,n\}.
\ee
An element $p$ of ${\cal W}_n$ is thus a linear combination (with complex coefficients) of words in the $2n$ letters $\bar a$ and ${\bar a}^*$. The words and elements of ${\cal W}_n$ can also be viewed as monomials and polynomials, respectively, in the $2n$ variables $\bar a$ and ${\bar a}^*$.

Using the CCRs, any element $p$ of a Weyl algebra can be brought to the \emph{normal form}
\be
p=\sum_{\bar{s},\bar{t}}p_{\bar{s},\bar{t}}\,({a^*})^{\bar{s}}a^{\bar{t}},
\label{normal}
\ee
where $\bar s =(s_1,\ldots,s_n),\bar t =(t_1,\ldots,t_n)\in \N^n$, $a^{\bar t}$ denotes the monomial $a^{\bar t}=\prod_{i=1}^na_{i}^{t_i}$ and similarly $(a^*)^{\bar s}=\prod_{i=1}^n{a^*_{i}}^{s_i}$,  and where $p_{\bar{s},\bar{t}}\in \C$ are complex coefficients. For instance, the elements $a_1a_1^*$ and $a_1a_2a_2^*a_1^*$, expressed in normal form, are $1+a_1^*a_1$ and $1+a_1^*a_1+a_2a_2^*+a^*_1a^*_2a_1a_2$, respectively. We will later show that the decomposition (\ref{normal}) is unique for each $p\in {\cal W}_n$. To check whether two different polynomials $p,p'$ in the variables $\bar{a},\bar{a}^*$  represent the same element $p=p'\in {\cal W}_n$ it is therefore enough to write $p,p'$ in normal form. 

This last observation suggests a natural norm in $\W_n$: let $p\in\W_n$, and let (\ref{normal}) be its normal decomposition. Then, we define the $l_1$-norm of $p$ as 
\be
l_1(p)\equiv \sum_{\bar{s},\bar{t}}|p_{\bar{s},\bar{t}}|.
\ee
It is also useful to distinguish the elements of ${\cal W}_n$ by resorting to the concept of degree. We say that the degree of an element $p$ of ${\cal W}_n$ is $\deg(p)=\max\{\|\bar{s}\|_1+\|\bar{t}\|_1:p_{\bar{s}\bar{t}}\not=0\}$ for $p$ expressed in normal form (\ref{normal}). It is easy to see that the degree of any monomial $s=\prod_{k=1}^d s_k$ with $s_k\in\{a_j,a_j^*\}_{j=1}^n$ is equal to $\deg(s)\equiv|s|=d$.

Given a monomial $s$, we denote by $\ddagger s\ddagger$ its \emph{anti-normal ordering}, that is, the monomial $s'$ that results when we reorder the letters appearing in the expression of $s$ in such a way that all the letters $a_1,a_2,...,a_n$ end up on the left and all the letters ($a^*_1,a^*_2,...,a^*_n$) end up on the right. For example, $\ddagger a_2^*(a^*_1)^2a_1\ddagger=\ddagger a^*_1a_1a_1^*a_2^*\ddagger=a_1(a^*_1)^2a_2^*$.

In this article, we will be mainly concerned with \emph{hermitian} elements of $\W_n$, i.e., those polynomials $p$ such that $p=p^*$. If $p$ is decomposed as in (\ref{normal}), the hermiticity condition thus translates as $p_{\bar{s},\bar{t}}=p_{\bar{t},\bar{s}}^*$, where $p_{\bar{t},\bar{s}}^*$ denotes the complex conjugate of the number $p_{\bar{t},\bar{s}}$.

\subsection{The Schr\"odinger representation}
Let $\pi$ be a mapping $\pi:{\cal W}_n\to L(\H)$, for some Hilbert space $\H$, where $L(\H)$ denotes the space of linear (not necessarily bounded) operators of $\H$. We say that $\pi$ is a $*$-representation of ${\cal W}_n$ if and only if
\begin{enumerate}
\item
$\pi(1)=\id_\H$.
\item
$\pi(pq)=\pi(p)\pi(q)$, $\pi(p+q)=\pi(p)+\pi(q)$ for all $p,q\in \W_n$.
\item
$\pi(p)^*=\pi(p^*)$, for all $p\in \W_n$.
\end{enumerate}
We now show that ${\cal W}_n$ admits a representation. For this, let $H$ be a separable Hilbert space, and let $\{\ket{s},s\in\N\}$ be an orthonormal basis for $H$, which we call the \emph{number basis}. If we denote by $\tilde{a}\in L(H)$ the linear operator defined by 
\be 
\tilde{a}\ket{s}=\left\{
\begin{array}{ll}0& \mbox{ for } s=0,\nonumber\\
\sqrt{s}\ket{s-1} & \mbox{ otherwise.}
\end{array}\right.
\ee
then its adjoint $\tilde{a}^*$ satisfies
\be
\tilde{a}^*\ket{s}=\sqrt{s+1}\ket{s+1},
\ee
and so it can be verified that 
\be
[\tilde{a},\tilde{a}^*]\ket{s}=\ket{s},\mbox{ for all }s\in \N.
\ee
Defining $\H=H^{\otimes n}$, we can then build a representation $\pi_S:{\cal W}_n\to L(\H)$ for the Weyl algebra ${\cal W}_n$ through
\be
\pi_S(a_k)=\id^{\otimes k-1}\otimes\tilde{a}\otimes\id^{\otimes n-k}.
\ee
This representation of $\W_n$ is known as the \emph{Schr\"{o}dinger (or Fock) representation}. From now on, we always refer to this representation and write $\pi$ for $\pi_S$ for simplicity.

\subsection{Weyl polynomial minimization}
The Schr\"odinger representation admits a clear physical interpretation: given a set of $n$ one-dimensional particles, it associates to each particle $k\in\{1,...,n\}$ a pair of \emph{creation and annihilation operators} $\pi(a_k), \pi(a^*_k)$. The operators describing the position and momentum of particle $k$ along the real line are then given, respectively, by $\pi(x_k)\equiv\pi(\frac{a_k+a^*_k}{\sqrt{2}})$ and $\pi(p_k)\equiv\pi(\frac{a_k-a^*_k}{\sqrt{2}i})$.
If these particles are subject to a potential $V(x_1,...,x_n)$, the energy operator of the system, in non-relativistic approximation, will be given by
\be
\pi(E)=\pi\left(\sum_{i=1}^n\frac{p_i^2}{2m_i}+V(x_1,...,x_n)\right),
\ee
where $m_i\in\R^+$ is the mass of particle $i$. In particular, the minimum energy of the system will be given by
\be
\lambda_{\inf}(E)=\inf\,\{\bra{\phi}\pi(E)\ket{\phi}:\ket{\phi}\in \mathcal{S},\braket{\phi}{\phi}=1\},
\ee
where $\mathcal{S}$ is the Schwartz space, that is, the set of states $|\phi\rangle$ (i.e., vectors of ${\cal H}$) satisfying $\|\pi(p) \phi\rangle\|<\infty$ for all $p\in \mathcal{W}_n$. Note that the minimization over $\bra{\phi}E\ket{\phi}$ makes sense, because the energy is an hermitian operator, $E=E^*$, and  consequently, $\bra{\phi}E\ket{\phi}\in\R$ is a real quantity for all $|\phi\rangle\in \mathcal{S}$. 

The case where $E$ is a polynomial in the variable $x_i,p_j$ -- or equivalently in the variables $a_i,a^*_j$ -- is particularly important (it includes for instance the case where the potential $V(x_1,...,x_n)$ is Taylor expanded around some equilibrium position). This motivates the following generic Weyl polynomial optimization problem 
\be
\lambda_{\inf}(p)=\inf\,\{\bra{\phi}\pi(p)\ket{\phi}:\ket{\phi}\in \mathcal{S},\braket{\phi}{\phi}=1\},
\ee
for an arbitrary hermitian polynomial $p\in {\cal W}_n$.
Note that, alternatively, we can write

\be \label{min}
\lambda_{\inf}(p)=\sup\,\{\lambda\in\R:\pi(p)-\lambda\geq 0\},
\ee
where positivity is understood in the Schwartz space. This reformulation of the problem will be used in the next section.

\subsection{Coherent states}
An interesting subset of $\H$ is constituted by the \emph{coherent states}. For any $\alpha\in\C$, denote by $\ket{\alpha}\in H$ the normalized state 
\be
\ket{\alpha}\equiv e^{-\frac{|\alpha|^2}{2}}\sum_{s\in \N}\frac{\alpha^s}{\sqrt{s!}}\ket{s}.
\ee
Then, a coherent state in $\H$ is any state of the form $\ket{\bar{\alpha}}=\otimes_{i=1}^n\ket{\alpha_i}$, for any $\bar{\alpha}\in\C^n$. Coherent states are important because they are simultaneous eigenstates of the annihilation operators $\{\pi(a_k):k=1,...,n\}$. This follows from the easily verified identity $\tilde{a}\ket{\alpha}=\alpha\ket{\alpha}$ for all $\alpha\in \C$. As a result, for any normally-ordered polynomial $p\in \W_n$, we have that
\be
\bra{\bar{\alpha}}\pi(p)\ket{\bar{\alpha}}=\sum_{\bar{s},\bar{t}}p_{\bar{s},\bar{t}}\,{\alpha^*}^{\bar{s}}\alpha^{\bar{t}}\equiv p(\bar{\alpha},\bar{\alpha}^*).
\ee
As an application of the coherent states, let us show that the decomposition (\ref{normal}) is unique, or, equivalently, that $0$ admits a unique representation. Suppose thus that 
\be
0=\sum_{\bar{s},\bar{t}}p_{\bar{s},\bar{t}}\,{a^*}^{\bar{s}}a^{\bar{t}},
\ee
for some coefficients $p_{\bar{s},\bar{t}}\in\C$. Using the Schr\"{o}dinger representation, we have that 
\be
0=\bra{\bar{\alpha}}\pi(0)\ket{\bar{\alpha}}=\sum_{\bar{s},\bar{t}}p_{\bar{s},\bar{t}}\,{\alpha^*}^{\bar{s}}\alpha^{\bar{t}}\quad \mbox{for all }\bar\alpha\in \C^n.
\ee
The right-hand side is a polynomial in the complex variables $\bar{\alpha},\bar{\alpha}^*$, and it can only be equal to zero for all values of $\bar{\alpha}$ if $p_{\bar{s},\bar{t}}=0$ for all $\bar{s},\bar{t}$.

\section{SOS decompositions, the SDP hierarchy, and the paradox}
\label{method}
Given a hermitian polynomial $p\in\W_n$, we say that $p$ admits a sum-of-squares (SOS) decomposition if there exist polynomials $f_i\in \W_n$ such that
\be
p=\sum_i f_i^*f_i.
\label{cuadrado}
\ee
We denote $\Sigma^2$ the set of all such polynomials.
It is clear that if $p\in\Sigma^2$, then $\pi(p)\geq 0$, since, for any $\ket{\phi}\in {\cal S}$,
\be
\bra{\phi}\pi(p)\ket{\phi}=\sum_i\bra{\phi} \pi(f_i^*)\pi(f_i)\ket{\phi}\geq 0.
\ee
However, the opposite implication is not true, not even in $\W_1$. Indeed, as noted by Schm\"{u}dgen \cite{schmudgen}, the family of polynomials $p_\epsilon=(a_1^*a_1-1)(a_1a_1^*-2)+\epsilon$ satisfies $\pi(p_\epsilon)\geq 0$, for $\epsilon\geq 0$, but nevertheless $p_\epsilon\not\in\Sigma^2$, for $\epsilon<\frac{1}{4}$.

Given a $p\in\W_n$, a possible scheme for finding a lower bound $\lambda^\star$ on the solution $\lambda_{\inf}(p)$ of (\ref{min}) is thus to solve the problem
\be\label{SOS_scheme}
\lambda^\star\equiv\sup\{\lambda\in \R: p-\lambda\in \Sigma^2\}.
\ee
This principle is the one behind the polynomial minimization algorithms developed by Lasserre \cite{moment} and Parrilo \cite{parrilo}, and their non-commutative analogue \cite{POI,chapter}. Such algorithms work by searching for SOS decompositions of $p-\lambda$ with some degree constraint. Applied to the Weyl minimization problem, this results in the following sequence of programs:
\be
\lambda^k\equiv\sup\{\lambda\in \R: p-\lambda\in \Sigma^2_k\}.
\label{SOS_formal}
\ee
Here $k$ is an integer such that $2k\geq \deg(p)$ and indexing the successive programs in the sequence and $\Sigma^2_k$ is the set of polynomials which admit a decomposition of the form (\ref{cuadrado}) with $\deg(f_i)\leq k$, for all $i$. Each of these problems is a semidefinite program, as one can check that 
\be
\lambda^k=\max\{\lambda\in \R: p-\lambda=(\bar{w}^k)^*Z\bar{w}^k, Z\geq 0\},
\label{SOS_SDP}
\ee
where $\bar{w}^k$ is a vector whose components are the normally ordered monomials of degree $\leq k$.
Clearly, $\lambda^k\leq \lambda^{k+1}$ and $\lim_{N\to\infty}\lambda^k=\lambda^\star$.
The programs (\ref{SOS_formal}), (\ref{SOS_SDP}) thus form a converging hierarchy of SDP relaxations for the problem (\ref{SOS_scheme}). Supplemented with a boundedness condition (that is not satisfied in the present case of Weyl polynomials), it can be further be shown that this hierarchy necessarily converges to the optimal solution of the problem (\ref{min}) \cite{POI,chapter}, as problems (\ref{SOS_scheme}) and (\ref{min}) then turn out to be equivalent \cite{positivsten}.

Let $L$ denote an arbitrary functional on the Weyl algebra, i.e., $L:\W_n\to\C$. Then, the dual of problems (\ref{SOS_formal}), (\ref{SOS_SDP}) can be shown to be
\be
\mu^k=\min\{L(p): L(1)=1,L(qq^*)\geq 0\mbox{ for all }q\in{\cal W}_n, \deg(q)\leq k\},
\label{moment_formal}
\ee
or explicitly in SDP form
\be
\mu^k=\min\{\sum_{\bar{s},\bar{t}} p_{\bar{s},\bar{t}}\,y_{\bar{s},\bar{t}}\,:y_{\bar{0},\bar{0}}=1, M_k(y)\geq 0 \},
\label{moment_SDP}
\ee
In this last formulation, $p_{\bar{s},\bar{t}}$ are the coefficients of $p$ in normal form (\ref{normal}),  $y\equiv\{y_{\bar{s},\bar{t}}:\|\bar{s}+\bar{t}\|_1\leq 2k\}\subset \C$ are the optimization variables\footnote{Note that, if $p_{\bar{s},\bar{t}}\in \R,\forall\bar{s},\bar{t}$, one can assume $\{y_{\bar{s},\bar{t}}:\|\bar{s}+\bar{t}\|_1\leq 2k\}\subset \R$.}, and $M_k(y)$ is the \emph{moment matrix of order $k$}, a matrix whose rows and columns are indexed by pairs of vectors $(\bar{s},\bar{t})\in \N^n\times\N^n$, with $\|\bar{s}+\bar{t}\|_1\leq k$, and with entries defined by
\be
\left[M_k(y)\right]_{(\bar{s},\bar{t}),(\bar{u},\bar{v})}=\sum_{\bar{q},\bar{r}}c_{\bar{q},\bar{r}}\,y_{\bar{q},\bar{r}},
\ee
where $\{c_{\bar{q},\bar{r}}\}$ are the normal form coefficients of the monomial  $({a^*}^{\bar{s}}a^{\bar{t}})^*{a^*}^{\bar{u}}a^{\bar{v}}$, i.e.,
\be
({a^*}^{\bar{s}}a^{\bar{t}})^*{a^*}^{\bar{u}}a^{\bar{v}}=\sum_{\bar{q},\bar{r}}c_{\bar{q},\bar{r}}{a^*}^{\bar{q}}a^{\bar{r}}.
\ee
The next lemma shows that problems (\ref{SOS_formal}),(\ref{SOS_SDP}) and (\ref{moment_formal}),(\ref{moment_SDP}) are, in fact, equivalent:
\begin{lemma}
If there exists a feasible point of program (\ref{SOS_SDP}), then $\mu^k=\lambda^k$.
\end{lemma}
\begin{proof}
By Sylvester's criterion \cite{sdp}, it just suffices to show that problem (\ref{moment_formal}) admits a strictly feasible point, i.e., that there exists a functional $L$ such that $L(1)=1$ and $L(pp^*)>0$, for all $p\not=0$. 
Now, consider the functional $L(p)=\tr(\Omega\,\pi(p))$, with
\be
\Omega= \frac{1}{(2\pi)^n}\int\! \mathrm{d}\bar{\alpha}\,e^{-|\bar{\alpha}|^2} \proj{\bar{\alpha}}.
\label{state}
\ee
This functional satisfies $L(1)=1$. Also, for any non-zero polynomial $p\neq 0\in \W_n$ we have that
\begin{eqnarray}
L(p^*p)&=&\frac{1}{(2\pi)^n}\int \! \mathrm{d}\bar{\alpha}\, e^{-|\bar{\alpha}|^2} \bra{\bar{\alpha}}\pi(p^*p)\ket{\bar{\alpha}}
\geq\frac{1}{(2\pi)^n}\int\! \mathrm{d}\bar{\alpha}\, e^{-|\bar{\alpha}|^2} |\bra{\bar{\alpha}}\pi(p)\ket{\bar{\alpha}}|^2\nonumber\\
&=&\frac{1}{(2\pi)^n}\int \! \mathrm{d}\bar{\alpha}\,e^{-|\bar{\alpha}|^2} |p(\bar{\alpha},\bar{\alpha}^*)|^2>0.
\end{eqnarray}
The last inequality comes from the fact that $p(\bar{\alpha},\bar{\alpha}')\not=0$ and that the integration takes place in all $\C^n$.
\end{proof}

Even though the above SDP hierarchy cannot be guaranteed to converge to the optimal value $\lambda_{\inf}(p)$ of (\ref{min}), every SDP step provides a lower-bound on  $\lambda_{\inf}(p)$. Based on the successful applications of this SDP hierarchy to fermionic systems and quantum correlations, where in practice very good lower bounds are obtained after only a few SDP relaxations, one could expect a good overall performance also in the context of Weyl polynomials. 
However, as the next result shows, no improvement over the first lower-bound can be obtained by considering higher steps in the hierarchy.
\begin{lemma}
\label{paradox1}
Let $p\in \Sigma^2$. Then, $p\in \Sigma^2_k$ for $2k\leq \deg(p)$.
\end{lemma}
\begin{proof}
By hypothesis, $p\in\Sigma^2$, and thus there exist polynomials $f_i$ such that (\ref{SOS}) holds. We will show that all such polynomials satisfy $2\deg(f_i)\leq \deg(p)$.

Suppose, on the contrary, that
\be
p=\sum f_i^*f_i+f^*f,
\ee
with $2\deg(f)>\deg(p)$. Then,
\be
p(\bar{\alpha},\bar{\alpha}^*)=\bra{\bar{\alpha}}\pi(p)\ket{\bar{\alpha}}\geq \bra{\bar{\alpha}}\pi(f^*)\pi(f)\ket{\bar{\alpha}}\geq |\bra{\bar{\alpha}}\pi(f)\ket{\bar{\alpha}}|^2=|f(\bar{\alpha},\bar{\alpha}^*)|^2.
\label{intringulis}
\ee
Now, denote by $LT(f)$ the leading terms of $f$, i.e.,
\be
LT(f)(\bar{\alpha},\bar{\alpha}^*)=\sum_{\|\bar{s}+\bar{t}\|_1=\deg(f)}f_{\bar{s},\bar{t}}(\alpha^{\bar{s}})^*\alpha^{\bar{t}},
\ee
and choose $\bar{\beta}\in\C^n$ such that $LT(f)(\bar{\beta},\bar{\beta}^*)=c\not=0$. Then it is straightforward that $|f(r\bar{\beta},r\bar{\beta^*})|=O(r^{\deg(f)})$. However, $|p(r\bar{\beta},r\bar{\beta}^*)|\leq O(r^{\deg(p)})$. For (\ref{intringulis}) to hold, we must thus have that $2\deg(f)\leq \deg(p)$.
\end{proof}

What Lemma \ref{paradox1} shows is that, for any polynomial $p$, the sequence of values $\lambda^{k_0},\lambda^{k_0+1},...$, with $k_0=\lceil\deg(p)/2\rceil$ is constant and equal to $\lambda^\star$. In other words: the first SDP relaxation of the problem already provides the best approximation to $\lambda_{{\inf}}(p)$ attainable with SOS decompositions.

How does such an approximation perform? Consider the uniparametric family of one-dimensional hamiltonians $\{E_m:m\in\R\}$, with
\be
E_m=\frac{1}{2}\tilde{p}^2+m\tilde{x}^2+\tilde{x}^4,
\ee
with $\tilde{x}\equiv (\tilde{a}+\tilde{a}^*)/\sqrt{2}$, $\tilde{p}\equiv (\tilde{a}-\tilde{a}^*)/i\sqrt{2}$. Note that, for $m<0$, $E_m$ corresponds to the interesting double-well potential. Figure \ref{first_rel} shows a plot of $\lambda^2$ as a function of $m$, together with an upper bound on $E_m$ obtained through variational methods. We used the solver SDPT3-4.0 \cite{SDPT3} and the MATLAB package YALMIP \cite{yalmip} to carry out the SDP calculations. It is clear that, as soon as $m<0$, the approximation given by $\lambda^2$ becomes worse.
\begin{figure}
  \centering
  \includegraphics[width=12 cm]{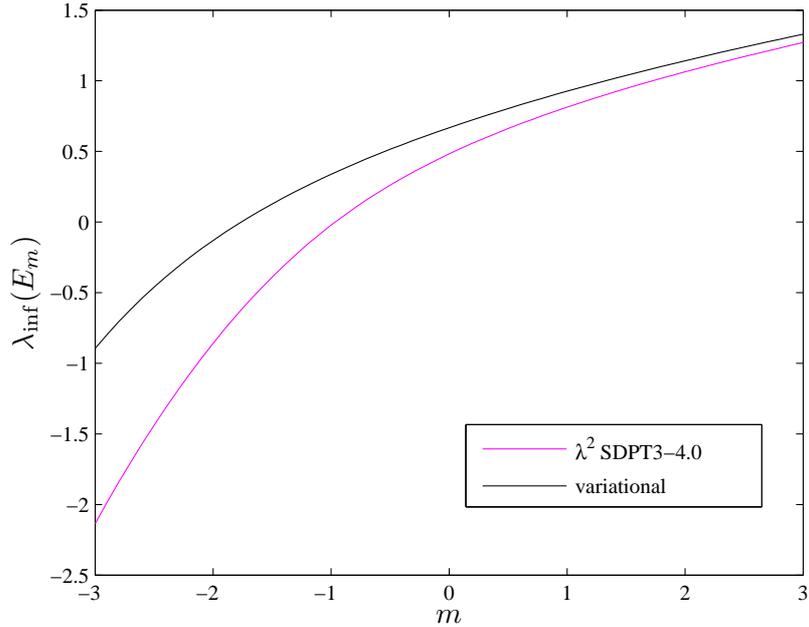}
  \caption{Bounds on $\lambda_{\inf}(E_m)$ as a function of $m$. The upper curve represents an upper bound obtained through variational techniques. The lower curve represents the lower bound $\lambda^2$ corresponding to the first relaxation of the problem.}
  \label{first_rel}
\end{figure}

From our discussions above, it follows that $\lambda^2=\lambda^3=\lambda^4=...$, and so $\lambda^2$ represents the best lower bound on $\lambda_{\inf}(e_m)$ achievable with the SDP hierarchy. Figure \ref{all_rel} shows, however, that such is not the case. Indeed, we see that subsequent relaxations of the problem return lower bounds which are closer and closer to the variational upper bound, until, at $\lambda^6$, both bounds become practically indistinguishable. What is happening?
\begin{figure}
  \centering
  \includegraphics[width=12 cm]{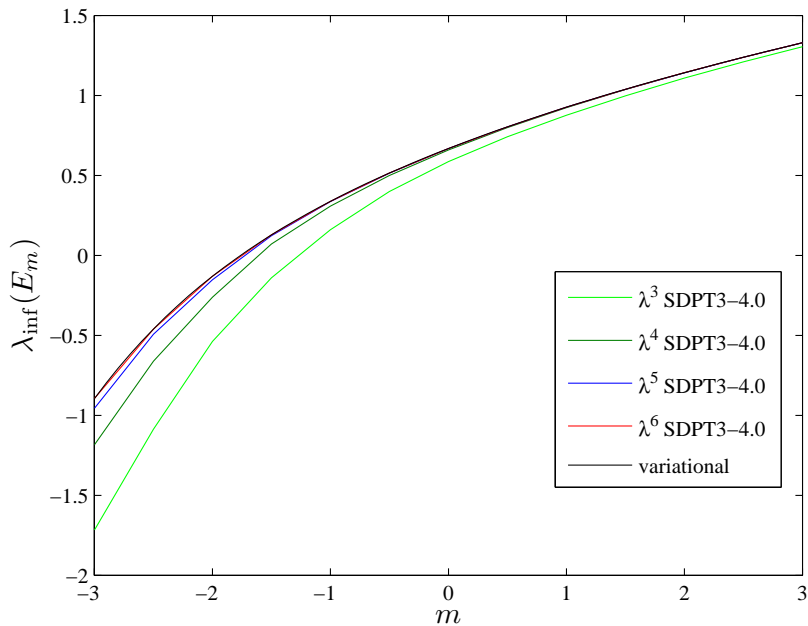}
  \caption{Lower-bounds on $\lambda_{\inf}(E_m)$ as a function of $m$ corresponding to the second, third, fourth and fifth SDP relaxations, in ascending order. The black line represents the variational estimate, quite close to $\lambda^6$ (red line).}
  \label{all_rel}
\end{figure}

\section{Resolution of the paradox}
\label{positivstellensatz}
As mentioned in the introduction, the above paradox is not new in commutative polynomial optimization: indeed, Henrion and Lasserre \cite{henrion} noticed that the numerical implementation of their SDP algorithm for polynomial minimization returned the optimal value of the 2-dimensional Motzkin polynomial, instead of its SOS value ($-\infty$). Lasserre successfully solved this paradox by proving that any commutative positive polynomial can be approximated arbitrarily well by an SOS decomposition \cite{approx}. The accepted resolution of the paradox was that the rounding errors occurring during the numerical computations perturbed the polynomial $p$ to be minimized to another one of higher degree admitting an SOS decomposition \cite{waki}.

In this Section we will prove a non-commutative analog of this result, namely, that Weyl polynomials which are positive semidefinite in the Schr\"{o}dinger representation can be perturbed to a higher degree polynomial in $\Sigma^2$. This is formally stated in the following Theorem: 
\begin{theo}\label{principal}
Let $p$ be an element of ${\cal W}_n$ such that $\pi(p)\geq 0$. Then, for any $\epsilon,\delta>0$, there
exists a polynomial $\tilde{p}\in \Sigma^2$ such that $l_1(p-\tilde{p})<\epsilon$
and $|\lambda_{\inf}(\tilde{p})-\lambda_{\inf}(p)|<\delta$.
\end{theo}
The proof of this theorem follows straightforwardly from the next three lemmas. In these lemmas, the constant $c$ is arbitrary but fixed to be $c>2$.
\begin{lemma}\label{SOS}
Let $p\in \W_n$ be a polynomial such that $\pi(p)\geq 0$. Let 
\begin{equation}
g^r_c\equiv \sum_{\|\bar{t}\|_1\leq
r} \frac{(n-1)!}{c^{\|\bar{t}\|_1}(n+\|\bar{t}\|_1-1)!}a^{\bar{t}}(a^{\bar{t}})^*,
\label{gs}
\end{equation} 
where the sum runs over all vectors $\bar{t}\in\N^n$ of length $\|\bar{t}\|_1$ less or equal than $r$. Then for any $\epsilon>0$, there exists some $r_0$ such that for all $r>r_0$,
\be
\tilde{p}^r=p+\epsilon g^r_c\in \Sigma^2.
\ee
\end{lemma}
\begin{lemma}\label{l1p} Let $\tilde{p}^r$ be defined as in Lemma \ref{SOS}. Then
\begin{equation}
l_1(p-\tilde{p}^r)\leq\epsilon\cdot\frac{c}{c-2},
\end{equation}
for any $r>0$.
\end{lemma}
\begin{lemma}\label{minimum}
Let $p$ be a hermitian polynomial in ${\cal W}_n$ such that $\lambda_{\inf}(p)$ exists, and let $g^r_c$ be defined as in (\ref{gs}). Then, for any $\delta>0$, there exists a number $q\in\R^+$ such that 
\be
\lambda_{\inf}(p) \leq \lambda_{\inf}(p+\epsilon g^r_c)\leq\lambda_{\inf}(p)+\delta+\epsilon q,
\ee

\noindent for all $r,\epsilon>0$.
\end{lemma}
The proofs of Lemma \ref{SOS} and \ref{minimum} are given here below. The proof of Lemma \ref{l1p} follows from the results presented in Appendix \ref{distancia}. This Lemma implies that $p$ and $\tilde{p}^r$ can be made arbitrarily close
if we take $\epsilon$ small enough. Moreover, if we express $p-\tilde{p}^r$ in the anti-normal form, i.e., $p-\tilde{p}^r=\sum_{\bar{s},\bar{t}}p^a_{\bar{s},\bar{t}}a^{\bar{s}}(a^{\bar{s}})^*$, the corresponding natural norm $l_1^a(p-\tilde{p}^r)\equiv\sum_{\bar{s},\bar{t}}|p^a_{\bar{s},\bar{t}}|$ is trivially bounded by $\epsilon\cdot e^{1/c}$. This implies that any computer implementation of program (\ref{SOS_SDP}) where the polynomial to minimize is expressed in normal or anti-normal form will require a lot of precision in order to distinguish $\tilde{p}^r$ from $p$ for low values of $\epsilon$.

\vspace{10pt}
\noindent\emph{Proof of Lemma \ref{SOS}.} The demonstration of this Lemma will make use of two lemmas, proven in Appendices \ref{antinormal}, \ref{states_sch}, respectively.

\begin{lemma}
\label{posi_anti}
Let $s\in {\cal W}_n$ be a monomial. Then,
\be
\ddagger ss^*\ddagger-ss^*\in \Sigma^2.
\ee

\end{lemma}

\begin{lemma}\label{represent}
Let $L$ be a linear functional in ${\cal W}_n$. If $L(h^* h)\geq 0$ for any $h\in {\cal W}_n$ and there exist $c,d>0,k\in \N$ such that for any monomial $s\in \W_n$, the relation 

\begin{equation}
|L(s)|\leq dc^{|s|}\Gamma(\frac{|s|+k+1}{2})
\end{equation}

\noindent holds, then there exists a non normalized quantum state $\rho$ (i.e., a non-negative, trace-class operator) such that $L(h)=\tr(\rho \pi(h))$, for any polynomial $h\in \W_n$.

\end{lemma}
Let us now proceed with the proof of Lemma~\ref{SOS}. 
Following Lasserre et al. \cite{Lasserre}, let $\W_n(r)$ be the set of elements of ${\cal W}_n$ with degree less or equal than $r$, and consider the semidefinite program:

\begin{equation}
\epsilon_r^\star=\mbox{min }_L\{ L(p)|L:\W_n(2r)\to\R, L \mbox{ linear }, L(g^r_c)\leq
1, L(h^* h)\geq 0,\forall h\in \W_n(r)\}.
\label{primal}
\end{equation}

\noindent Noting that $L=0$ is an admissible linear functional, we have that the problem has feasible points and that $\epsilon_r^*\leq 0$.  Condition $L(g^r_c)\leq 1$, together with Lemma \ref{posi_anti}, implies that the diagonal entries of all feasible moment matrices $M_r(y)$ are upper bounded, and so the absolute values of the rest of the entries, due to positive semidefiniteness.
From these two observations, it follows that our problem admits a solution, i.e., $\epsilon_r^\star\not=-\infty$ is attainable for a
feasible choice of $L$. 

The dual of (\ref{primal}) is

\begin{equation}
\mbox{max }_\epsilon\{-\epsilon: \epsilon\geq 0, p+\epsilon g^r_c\in \Sigma^2\}.
\end{equation}

\noindent That this problem has solutions for any $p=p^*$ follows from the fact that, for all $\mu\in \C$ and any pair of monomials $s,t\in \W_n$,

\be
\mu s^*t+\mu^*t^*s+|\mu|\{s^*s+t^*t\}\in \Sigma^2.
\ee

\noindent And thus, invoking Lemma \ref{posi_anti}, 

\be
\mu s^*t+\mu^*t^*s+|\mu|\{\ddagger s^*s\ddagger +\ddagger t^*t\ddagger\}\in \Sigma^2.
\ee

\noindent By increasing the value of $\epsilon$, at some point we will therefore have that $p+\epsilon g^r_c\in\Sigma^2$.

Moreover, in this particular case, there is no duality gap, i.e., the solutions of
both the primal and dual problems coincide. Again, this can be established by invoking the quantum state (\ref{state}): choosing $\sigma>0$ such that $\tr\{\Omega \pi(g^r_c)\}< 1/\sigma$, it follows that $L(h)\equiv\sigma\tr(\Omega\pi(h))$ is a strictly feasible point of (\ref{primal}) and, thus, the solutions of the dual and primal problems are the same \cite{sdp}.

This, together with the fact that $g^s_c-g^r_c$ is a sum of squares for $s\geq r$, implies that, for all $\epsilon\geq-\epsilon_r^\star\geq 0$, the polynomial

\begin{equation}
p+\epsilon g^s_c\mbox{ }(s\geq r)
\end{equation}

\noindent is also a sum of squares.

The sequence $(\epsilon^\star_r)_r$ is, therefore, and increasing one. We will
next proof that $\lim_{r\to\infty} \epsilon^\star_r=0$, and so that the $\epsilon>0$
appearing in the formulation of Lemma~\ref{SOS} can be taken arbitrarily
small.

Consider the sequence $L^\star_r$ of functionals that attain the solutions $\epsilon^\star_r$
of the problem, and denote by $M^\star_r$ their corresponding moment matrices (the entries
$(M^\star_r)_{(\bar{s},\bar{t}),(\bar{u},\bar{v})}$ where either $\|\bar{s}+\bar{t}\|_1> r$ or $\|\bar{u}+\bar{v}\|_1> r$ are assumed to be
completed with zeros). By Lemma \ref{posi_anti}, we have that 

\be
(M^\star_r)_{(\bar{s},\bar{t}),(\bar{s},\bar{t})}\leq (M^\star_r)_{(\bar{s}+\bar{t},0),(\bar{s}+\bar{t},0)}\leq \frac{(n+\|\bar{s}+\bar{t}\|_1-1)!}{(n-1)!}c^{\|\bar{s}+\bar{t}\|_1}=:d(\bar{s},\bar{t}).
\ee

\noindent Now, perform the transformation 

\be
(M^\star_r)_{(\bar{s},\bar{t}),(\bar{u},\bar{v})}\to (N^\star_r)_{(\bar{s},\bar{t}),(\bar{u},\bar{v})}\equiv d(\bar{s},\bar{t})^{1/2}(M^\star_r)_{(\bar{s},\bar{t}),(\bar{u},\bar{v})}d(\bar{u},\bar{v})^{1/2}. 
\ee

\noindent $N^\star_r$ is thus positive semidefinite and its diagonals are upper bounded by 1 for all $r$; it follows that all the entries of the matrices $N^\star_r$ are in the interval $[1, -1]$. By the Banach-Alaoglu theorem, the sequence $N^\star_r$ admits a subsequence $\{r_i\}$ that converges in the weak-$\ast$ topology to a limit $N^\star_{r_i}\to\hat{N}$ when $i\to\infty$ \cite{reedsimon}. Undoing the previous change of coordinates, we are left with an infinite sized matrix $\hat{M}$ that defines a linear functional $\hat{L}(p)\equiv\sum_{\bar{s},\bar{t}}p_{(\bar{s},\bar{t})}\hat{M}_{(0,\bar{s}),(0,\bar{t})}$ on the Weyl algebra. 

This functional satisfies $\hat{L}(h^*h)\geq 0$, for any polynomial $h$. Moreover, for any sequence $s=s_1s_2$, with $|s_1|=\left\lceil\frac{|s|}{2}\right\rceil$, $|s_2|=\left\lfloor \frac{|s|}{2}\right\rfloor$, 

\begin{eqnarray}
|\hat{L}(s)|&&\leq \left(\hat{L}(s_1s_1^*)\hat{L}(s_2^*s_2)\right)^{1/2}\leq\left(\frac{(n+|s_1|-1)!(n+|s_2|-1)!}{(n-1)!(n-1)!}\right)^{1/2}\sqrt{c}^{|s_1|+|s_2|}\nonumber\\
&&\leq\frac{\sqrt{c}^{|s|}}{(n-1)!}\Gamma\left(\frac{2n+2+|s|}{2}\right).
\end{eqnarray}

\noindent By Lemma \ref{represent}, this last condition implies that there exists a non normalized quantum state $\rho\in S_1(\H)$ in the Schr\"{o}dinger representation such that $\hat{L}(h)=\tr(\rho \pi(h))$, for all $h$.

Now,

\begin{equation}
\lim_{r\to\infty}\epsilon^\star_{r}=\hat{L}(p)=\tr(\rho
\pi(p))\geq 0,
\end{equation}

\noindent where the last inequality follows from the non-negativity assumption
on $p$. On the other
hand, $\epsilon^\star_{r}\leq 0$ $\forall r$, and, therefore, we have that $\lim_{r\to\infty}\epsilon^\star_{r}=0$.\begin{flushright}$\square$\end{flushright}

\noindent\emph{Proof of Lemma \ref{minimum}.} Given a vector $\bar{m}\in \N^n$, we will denote by $\ket{\bar{m}}$ the number state $\ket{m_1}\otimes\ket{m_2}\otimes...\otimes\ket{m_n}$. Now, if $\pi(p)\geq 0$, then $\lambda_{\inf}(p)$ exists and can be written as
the limit of a sequence of the form $\bra{\phi_i}\pi(p)\ket{\phi_i}$, where $\{\phi_i\}$
are normalized quantum states. Such states can be, in turn, approximated with
arbitrary precision by finite linear combinations of number states. Choose,
then, a number $M\in \N$ such that the normalized state $\ket{\Phi}\equiv\sum_{\|\bar{m}\|_{\infty}\leq M} d_{\bar{m}} \ket{\bar{m}}$ satisfies $\lambda_{\inf}(p)\leq\bra{\Phi}\pi(p)\ket{\Phi}\leq\lambda_{\inf}(p)+\delta$. 

It can be verified that, for any pair of number states $\ket{\bar{m}},\ket{\bar{m}'}$, with $\|\bar{m}\|_{\infty},\|\bar{m}'\|_{\infty}\leq M$, and any monomial $s$ of the annihilation
and creation operators, the inequality $|\bra{\bar{m}}ss^*\ket{\bar{m}'}|\leq \frac{(|s|+M)!}{M!}\delta_{\bar{m},\bar{m}'}$ holds.
Therefore, 

\begin{equation}
\bra{\Phi}ss^*\ket{\Phi}\leq \frac{(|s|+M)!}{M!}\cdot\sum_{\bar{m}} |d_{\bar{m}}|^2=\frac{(|s|+M)!}{M!}.
\end{equation}

\noindent Finally, choose $c>1$. It follows that 

\be
\bra{\Phi}g^r_c\ket{\Phi}\leq \sum_{l=0}^\infty\frac{\sharp\{\bar{t}:|\bar{t}|\leq l\}}{\frac{(n+l-1)!}{(n-1)!l!}}\frac{(l+M)!}{M!l!c^l}=
\sum_{l=0}^\infty \frac{(l+M)!}{M!l!c^l}=\left(\frac{c}{c-1}\right)^{M+1},
\ee

\noindent where in order to identify the second and third expressions we made use of Proposition \ref{combina} in Appendix \ref{distancia}. Equaling to $q$ the last result, we arrive at the promised Lemma.\begin{flushright}$\square$\end{flushright}

Note that Lemma \ref{represent} alone provides an alternative explanation for the observed convergence to the optimal solution, this time from the point of view of the dual problem (\ref{moment_formal}). The reason why $\mu^k$ does not necessarily converge in theory to $\lambda_{\inf}(p)$ for the exact problem is because the values $L(s)$ in each program grow faster than $dc^{|s|}\Gamma(\frac{|s|+l}{2})$. If, however, due to finite numerical precision our solvers limit the magnitude of such momenta, low order relaxations $\mu^k$ should provide a better approximation to $\lambda_{\inf}$\footnote{This does not apply to high order relaxations corresponding to $k\gg 1$, as a constraint of the form $L(a_j^k(a_j^k)^*)<K$ becomes unimplementable as soon as $K>k!$.}. 

Increasing the numerical precision of our programs should thus have a two-fold effect: on one hand, it should allow the computer to distinguish between $p$ and its perturbation $\tilde{p}$. On the other hand, it should extend the moment matrix search space to include matrices with entries of very different magnitude. A high precision numerical computation should therefore make the curves in Figure \ref{all_rel} collapse to the same line.

We used the semidefinite programming solver SDPA-GMP \cite{nakata2,SDPA} to compute SDP approximations of $E_{-1},E_1$ with a precision of $600$ digits. Figures \ref{fig_minus1} and \ref{fig_plus1} show the outputs $\lambda^k-\lambda^2$ of both problems as a function of \emph{lambdaStar}, an internal parameter of SDPA-GMP that constrains the magnitudes of the entries of the moment matrix\footnote{More concretely, lambdaStar is such that $(\mbox{lambdaStar})\cdot\id-M_k(y)\geq 0$.} \cite{SDPA}. Notice that, in agreement with the above interpretation, the difference between $\lambda^k$ and $\lambda^2=\lambda^\star$ tends to zero as the constraints on the moment matrix disappear (right end of Figures \ref{fig_minus1} and \ref{fig_plus1}). Conversely, the solutions of higher order relaxations start to differ from $\lambda^2$ and become closer to the actual solution of the problem as we restrict the magnitude of the entries of the moment matrix (left end of Figures \ref{fig_minus1} and \ref{fig_plus1}).
\begin{figure}
  \centering
  \includegraphics[width=12 cm]{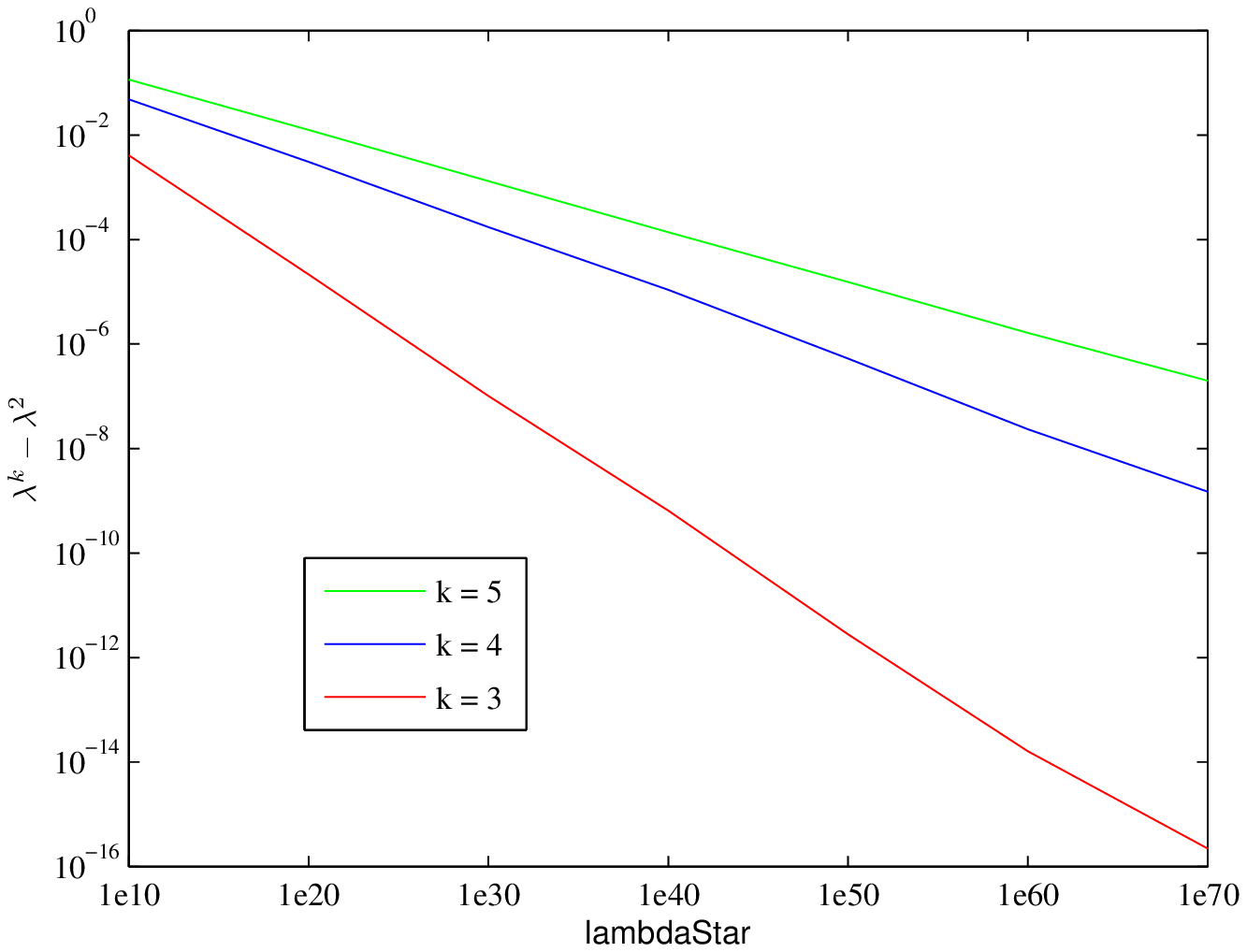}
  \caption{Plot of $\lambda^k-\lambda^2$ as a function of the parameter lambdaStar in the case $m=-1$.}
  \label{fig_minus1}
\end{figure}

\begin{figure}
  \centering
  \includegraphics[width=12 cm]{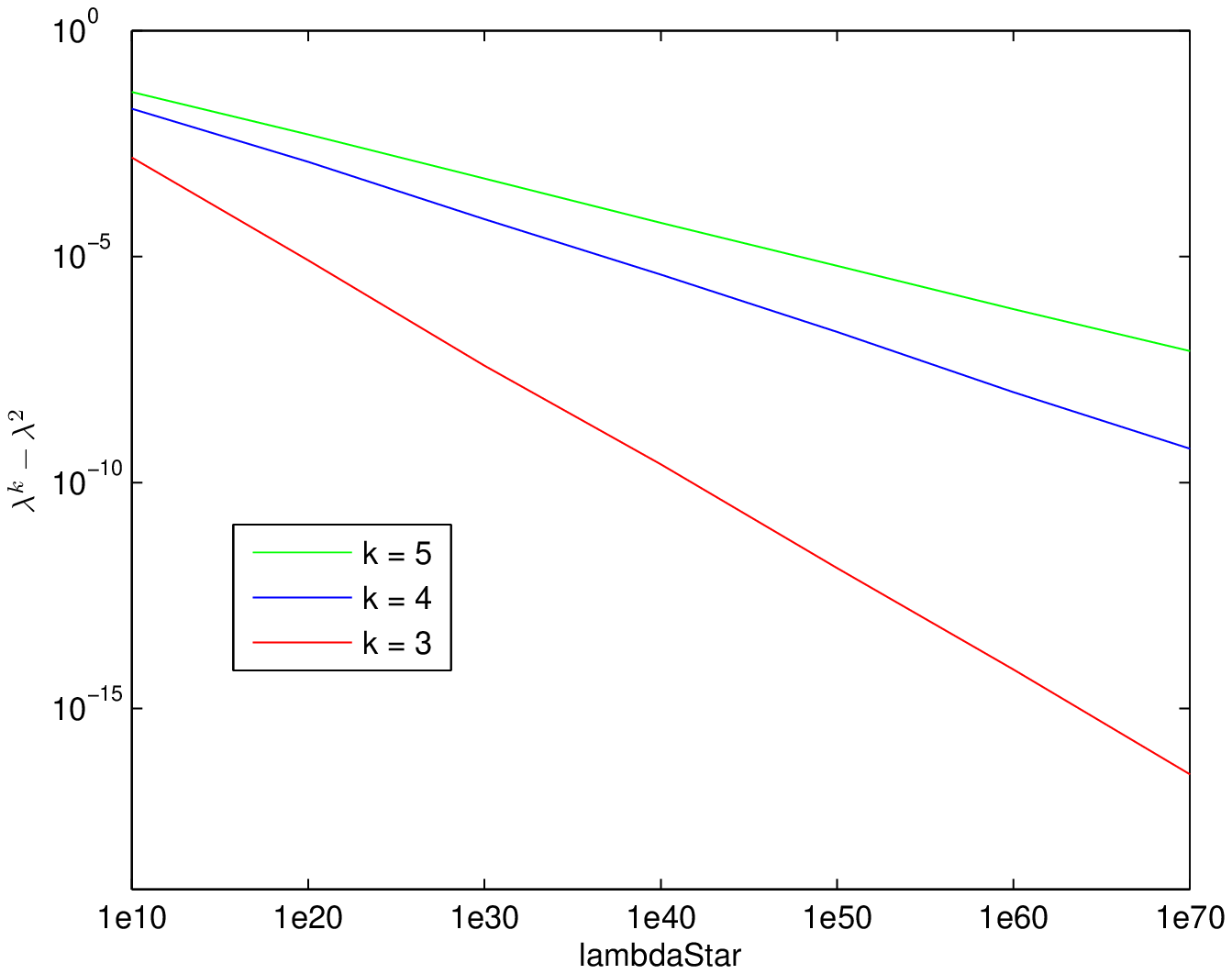}
  \caption{Plot of $\lambda^k-\lambda^2$ as a function of the parameter lambdaStar in the case $m=1$.}
  \label{fig_plus1}
\end{figure}

\section{Conclusion}
\label{conclusion}
We have identified a paradoxical behaviour in the application to bosonic systems of the SDP methods widely used in quantum chemistry energy calculations and quantum information. Namely, we have pointed out that numerical implementations of the method seems to converge despite a simple theoretical argument showing that the first SDP relaxation should already provide the best lower-bound to the problem at hand. This phenomenon is similar to an analogous  behavior observed in commutative polynomial optimization \cite{henrion,waki} and we suggested that the paradox arises from rounding errors introduced in the numerical comutation. We provided a theoretical basis for this assumption by proving that for any bosonic hamiltonian to be minimized there exists a perturbation of it whose ground state energy can accurately be approximated by the SDP method. Furthermore, we showed that the effect disappears as soon as we increase the computer precision. 

Our results suggest that the above problem could be avoided by constraining the values of the diagonal elements of the moment matrices in each program. Thanks to such constraints, computer implementations of the primal problem could return reliable solutions without having to resort to extremely high precision numerical calculations. This approach will be explored in a forthcoming article.   

It is worth noting that Cimpri\v{c} \cite{cimpric} proposed to use Schm\"{u}dgen's positivstellensatz for Weyl algebras \cite{schmudgen} to introduce a different SDP hierarchy than the one presented here in order to find rigorous lower bounds on the minimum value of arbitrary Weyl polynomials.  The application of this method, however, requires high precision SDP solvers.

Let us conclude with a problem for the Noncommutative Real Algebraic Geometry community. We have shown that the set of SOS polynomials is dense in the set of positive semidefinite elements of the Weyl algebra, i.e., we can approximate any polynomial which is positive semidefinite in the Schr\"{o}dinger representation by a SOS. It would be interesting to know if this kind of results also hold in other algebras important for quantum chemistry. For instance, if such an `approximation property' were also true in algebras containing coulombian elements of the type $1/|\bar{x}_i-\bar{x}_j|$, then we would be able to estimate electronic molecular energies without the need of introducing orbital basis sets.

\section{Acknowledgements}
M. N. has been supported by the Templeton Foundation. We acknowledge Monique Laurent and Frank Vallentin for interesting discussions. S.P. acknowledges financial support from the Brussels-Capital Region through a BB2B grant. MBP has been supported by the Alexander von Humboldt Foundation.

\section*{Appendix}
\begin{appendix}
\section{Bounding the norm of $g^r_c$}
\label{distancia}

The goal of this section is to prove the next lemma, from which Lemma~\ref{l1p} is a direct corollary.

\begin{lemma}
Let $c>2$. Then, 

\begin{equation}
l_1(g^r_c)\leq \frac{c}{c-2}.
\end{equation}
\label{boundito}
\end{lemma}

\noindent The proof of this lemma, we will rely on the next two propositions.

\begin{prop}
\label{cachonda}
\be
a^k(a^k)^*=\sum_{m=0}^k\frac{k!^2}{m!^2(k-m)!}(a^m)^*a^m.
\ee
\end{prop}

\begin{proof}
Clearly, $a^k(a^k)^*=\sum_{l,m}s_{l,m}(a^*)^la^m$. Evaluating the mean value of the Schr\"{o}dinger representation of both polynomials with respect to an arbitrary coherent state $\ket{\alpha}$, we have
that

\begin{equation}
\sum_{l,m}s_{l,m}(\alpha^*)^l\alpha^m=e^{-|\alpha|^2}\sum_{j=0}^\infty\frac{(j+k)!}{j!^2}|\alpha|^{2j}=:g(|\alpha|^2).
\end{equation}

\noindent It follows that 

\be
s_{l,m}=\delta_{l,m}\frac{1}{m!}\left.\frac{d^mg(x)}{dx^m}\right|_{x=0}.
\ee

\noindent Now, it can be proven, by induction, that

\be
\frac{d^mg(x)}{dx^m}=\frac{k!}{(k-m)!}e^{-x}\sum_{j=0}^\infty\frac{(j+k)!}{j!(j+m)!}x^j.
\ee
if $m\leq k$, while $\frac{d^mg(x)}{dx^m}=0$ if $m>k$.
\noindent The statement of the proposition follows from these two relations.

\end{proof}

\begin{prop}
\label{combina}
Let $\sharp\{\bar{t}\in \N^n:\|\bar{t}\|_1=k\}$ denote the number of elements of $\N^n$ satisfying $\|\bar{t}\|_1=k$. Then,
\be
\sharp\{\bar{t}\in \N^n:\|\bar{t}\|_1=k\}=\frac{(n+k-1)!}{(n-1)!k!}.
\ee
\end{prop}
\begin{proof}
Our aim is to compute the number of ways in which $k$ identical balls can be contained in $n$ different boxes. Clearly, any possible configuration can be represented uniquely by a sequence of $k$ dots ``$.$'' and $n-1$ bars ``$|$''. The number of balls $n_1$ in box 1 would then correspond to the number of dots on the left of the first bar; the number of balls $n_j$ inside box $j$, for $2\leq j\leq n-1$, to the number of dots between the $j-1^{th}$ and the $j^{th}$ bars; the number of balls in box $n$, to the number of dots on the right of the $n-1^{th}$ box. For instance, the configuration $n_1=2,n_2=0,n_3=1$ would be represented by ``$..||.$''.

It is elementary that the number of permutations of $n+k-1$ elements, out of which $n-1$ and $k$ are indistinguishable, is equal to $\frac{(n+k-1)!}{(n-1)!k!}$.

\end{proof}

\noindent\emph{Proof of Lemma \ref{boundito}.} 
Proposition \ref{cachonda} implies that 

\be
l_1(a^k(a^k)^*)=\sum_{m=0}^k\frac{k!^2}{m!^2(k-m)!}\leq \sum_{m=0}^k\frac{k!}{m!(k-m)!}k!=2^kk!.
\label{bongo}
\ee

\noindent The last expression is logarithmically superadditive, i.e., $\prod_{i}2^{k_i}k_i!\leq 2^kk!$, for all sets of natural numbers $\{k_1,k_2,...\}$ such that $\sum_ik_i=k$. It follows that the bound given by eq. (\ref{bongo}) also holds for $l_1(a^{\bar{t}}(a^{\bar{t}})^*)$, with $\bar{t}\in \N^n,\|\bar{t}\|_1=k$. We thus have that 

\begin{eqnarray}
l_1(g^r_c)&&\leq \sum_{|\bar{t}|\leq r}\frac{(n-1)!}{c^{\|\bar{t}\|_1}(n+\|\bar{t}\|_1-1)!}l_1(a^{\bar{t}}(a^{\bar{t}})^*)\leq \sum_{k=0}^r\frac{\sharp\{\bar{t}:\|\bar{t}\|_1\leq k\}}{\left(\begin{array}{c}n+k-1\\n-1 \end{array}\right)}\frac{2^k}{c^k}=\nonumber\\
&&=\sum_{k=0}^r\left(\frac{2}{c}\right)^k\leq\sum_{k=0}^\infty\left(\frac{2}{c}\right)^k=\frac{c}{c-2},
\end{eqnarray}

\noindent where in the third inequality we have made use of Proposition \ref{combina}.\begin{flushright}$\square$\end{flushright}

\section{Anti-normal ordered monomials}
\label{antinormal}
The following appendix establishes Lemma \ref{posi_anti}.
\begin{prop}
\label{posi_aaa}
For any $k\in \N$,

\be
a^{k+1}(a^{k+1})^*-a^*a^k(a^k)^*a\in \Sigma^2.
\label{caso_part}
\ee

\end{prop}

\begin{proof}
\noindent Using the CCRs, we have that

\be
a^*a^k(a^k)^*a=-ka^{k-1}(a^k)^*a-(k+1)a^k(a^*)^k+a^{k+1}(a^*)^{k+1}.
\ee

\noindent Now, by induction, we have that, for any $0\leq l\leq k$, $a^l(a^k)^*a^{k-l}\in \Sigma^2$. Indeed, for $l=0$ the result is obvious. Suppose now that the result holds for $l$. Then

\be
a^{l+1}(a^k)^*a^{k-l-1}=a^{l}a(a^k)^*a^{k-l-1}=ka^l(a^*)^{k-1}a^{k-l-1}+a^l(a^k)^*a^{k-l},
\ee

\noindent and the last expression belongs to $\Sigma^2$ by hypothesis.

It follows that

\be
a^{k+1}(a^{k+1})^*-a^*a^k(a^k)^*aa=(k+1)a^k(a^k)^*+ka^{k-1}(a^k)^*a\in\Sigma^2.
\ee

\end{proof}

\begin{prop}
\label{cognazo}
Let $s\in {\cal W}_1$ be an arbitrary monomial of length $k$. Then,
\be
a^k(a^k)^*-ss^*\in \Sigma^2.
\ee
\end{prop}

\begin{proof}
We will prove the proposition by induction. Suppose, thus, that the proposition holds for all monomials of length smaller or equal than $k$, and let $s$ be an arbitrary monomial with $|s|=k+1$. There are two possibilities:

\begin{enumerate}
\item
$s=a\tilde{s}$, with $|\tilde{s}|=k$. Then we have that

\be
a^{k+1}(a^{k+1})^*-ss^*=a(a^{k}(a^{k})^*-\tilde{s}\tilde{s}^*)a^*=\sum_i af_if^*_ia^*\in \Sigma^2.
\ee

\item
$s=a^*\tilde{s}$, with $|\tilde{s}|=k$. Then we have that

\be
a^{k+1}(a^{k+1})^*-ss^*=\{a^{k+1}(a^{k+1})^*-a^*a^k(a^k)^*a\}+\{a^*(a^k(a^k)^* -\tilde{s}\tilde{s}^*)a\}.
\ee

\noindent The first term between brackets is a SOS by Proposition \ref{posi_aaa}; the second term is a SOS due to the induction hypothesis.

\end{enumerate}

\noindent To complete the induction we also have to show that the proposition also holds for $k=1$. But this is trivial, since, in that case, $aa^*-ss^*$ equals $0$ ($1$), for $s=a$ ($s=a^*$).

\end{proof}

\noindent\textbf{Lemma \ref{posi_anti}.} \emph{Let $s\in {\cal W}_n$ be a monomial. Then,}
\be
\ddagger ss^*\ddagger-ss^*\in \Sigma^2.
\ee

\begin{proof} 
Proposition \ref{cognazo} already shows that the lemma holds for $n=1$. Now, suppose that the lemma holds for $n$, and let $s=tu\in {\cal W}_{n+1}$, with $t$ ($u$) being a word with the letters $a_1,...,a_{n},a^*_1,...,a^*_n$ ($a_{n+1},a_{n+1}^*$). Let $\bar{t}\in\N^n$ be such that $\ddagger tt^*\ddagger=a^{\bar{t}}(a^{\bar{t}})^*$. Then,

\be
\ddagger ss^*\ddagger-ss^*=a^{\bar{t}} a^{|u|}_{n+1}(a^{|u|}_{n+1})^*(a^{\bar{t}})^*-a^{\bar{t}} uu^*(a^{\bar{t}})^*+u\ddagger tt^*\ddagger u^*-utt^*u^*.
\ee

\noindent The first two terms on the right hand side admit a SOS decomposition due to Proposition \ref{cognazo}. The two remaining terms belong to $\Sigma^2$ because of the induction hypothesis.
\end{proof}

\section{States in the Schr\"{o}dinger representation}
\label{states_sch}
In this appendix, we demonstrate the following lemma.

\textbf{Lemma \ref{represent}}. \emph{Let $L$ be a linear functional in ${\cal W}_n$. If $L(h^* h)\geq
0$ for any $h\in {\cal W}_n$ and there exist $c,d>0,k\in \N$ such that for any monomial $s\in \W_n$, the relation} 
\begin{equation}
|L(s)|\leq dc^{|s|}\Gamma(\frac{|s|+k+1}{2})
\label{condition}
\end{equation}

\noindent \emph{holds, then there exists a non normalized quantum state $\rho$ (a non-negative trace class operator) such that}

\begin{equation}
L(h)=\mbox{\textup{tr}}(\rho \pi(h)),
\end{equation}

\noindent \emph{for any polynomial $h\in \W_n$.}

\begin{proof}
Suppose that, indeed, such a functional exists and define its characteristic
function $\chi(\bar{\xi})$ as

\begin{equation}
\chi(\bar{\xi})=\sum_{l=0}^\infty \frac{L((i\bar{\xi}\cdot\sigma \bar{R})^l)}{l!}=:L(e^{i\bar{\xi}\sigma \bar{R}}),
\end{equation}

\noindent where $\bar{R}\in\W_n^{2n}$ is the vector of polynomials 

\be
\bar{R}=\left(\frac{a_1+a_1^*}{\sqrt{2}},\frac{a_1-a_1^*}{i\sqrt{2}},...,\frac{a_n+a_n^*}{\sqrt{2}},\frac{a_n-a_n^*}{i\sqrt{2}}\right),
\ee

\noindent and $\sigma$ denotes the symplectic form, i.e., $\sigma=\oplus_{l=1}^n\left(\begin{array}{cc}0&1\\-1&0\end{array}\right)$.

Using (\ref{condition}), we have that

\begin{equation}
|\chi(\bar{\xi})|\leq d\sum_{l=0}^\infty \frac{(2n\hat{\xi}c)^l\Gamma(\frac{l+k+1}{2})}{l!},
\label{cosita}
\end{equation}

\noindent with $\hat{\xi}=\frac{1}{\sqrt{2}}\max_k \{|\xi_{2k+1}+i\xi_{2k+2}|\}$. That the last
series converges for any value of $\hat{\xi}$ follows from the relation

\begin{equation}
\int_{-\infty}^\infty |x|^ke^{-x^2}dx=\Gamma(\frac{k+1}{2}).
\end{equation}

\noindent This allows us to write

\begin{equation}
\sum_{l=0}^\infty \frac{a^l\Gamma(\frac{l+k+1}{2})}{l!}=\int_{-\infty}^\infty e^{a|x|}|x|^ke^{-x^2}dx,
\end{equation}

\noindent and the last integral converges for all $a\in \R$.

Now, define the operator

\begin{equation}
\rho\equiv \frac{1}{(2\pi)^n}\int d\bar{\xi} \chi(\bar{\xi}) W_{-\bar{\xi}},
\end{equation}

\noindent with $W_{\bar{\xi}}=e^{i\bar{\xi}\sigma \bar{R}}$ being the so called \emph{Weyl operator} \cite{petz}. 

\noindent We will next prove that, for any polynomial $h$, $\tr(\rho \pi(h))=L(h)$.

Because $\tr(W_{\bar{\xi}}W_{\bar{\eta}})=(2\pi)^n\delta(\bar{\xi}+\bar{\eta})$, it follows
that $\tr(\rho W_{\bar{\xi}})=\chi(\bar{\xi})=L(W_{\bar{\xi}})$. Moreover,
from the \emph{Weyl relations}

\begin{equation}
W_{\bar{\xi}} W_{\bar{\eta}}=e^{-i\bar{\xi}\sigma\bar{\eta}/2}W_{\bar{\xi}+\bar{\eta}},
\label{weylrel}
\end{equation}

\noindent it is immediate that $\tr(\rho W_{\bar{\xi}}W_{\bar{\eta}})=L(e^{-i\bar{\xi}\sigma\bar{\eta}/2}W_{\bar{\xi}+\bar{\eta}})\equiv
f(\bar{\xi},\bar{\eta})$.
We will now show that 

\begin{equation}
f(\bar{\xi},\bar{\eta})=\lim_{r\to\infty}f^r(\bar{\xi},\bar{\eta}),
\label{limitito}
\end{equation}

\noindent where $f^r(\bar{\xi},\bar{\eta})\equiv\sum_{l,m=0}^r
L(\frac{(i\bar{\xi}\cdot\sigma \bar{R})^l}{l!}\frac{(i\bar{\eta}\cdot\sigma \bar{R})^m}{m!})$.

Note that $f$ is analytic, and that $f^r$ corresponds to a sort of truncated Taylor expansion of $f$ with respect to the variables $\bar{\xi},\bar{\eta}$, where only those monomials $\xi^{\bar{s}}\eta^{\bar{t}}$ with $\|\bar{s}\|_1,\|\bar{t}\|_1\leq r$ are present. Indeed, if we take any number of derivatives of $\xi_i,\eta_j$ (no more that $r$ of each) on both sides of (\ref{weylrel}) and then evaluate on the point $\bar{\xi}=\bar{\eta}=0$, we will arrive at two polynomials $p_1,p_2\in W_n$ on each side. Since this is a general relation between operators, both polynomials must be the same modulo the canonical commutation relations. It follows that 

\be
\frac{\partial^m\partial^{m'}}{\prod_{k=1}^m\partial \xi_{i_k}\prod_{l=1}^{m'}\partial\eta_{j_l}} f(\bar{\xi},\bar{\eta})\mid_{\bar{\xi}=\bar{\eta}=0}=L(p_1)=L(p_2)=\frac{\partial^m\partial^{m'}}{\prod_{k=1}^m\partial \xi_{i_k}\prod_{l=1}^{m'}\partial\eta_{j_l}}  f^r(\bar{\xi},\bar{\eta})\mid_{\bar{\xi}=\bar{\eta}=0},
\ee

\noindent for $m,m'\leq r$. Since $f(\bar{\xi},\bar{\eta})$ is analytic in $\bar{\xi},\bar{\eta}$, the sum of the absolute values of the coefficients of its Taylor expansion must converge. This implies that $\lim_{r\to \infty}f(\bar{\xi},\bar{\eta})-f^r(\bar{\xi},\bar{\eta})\to 0$ for fixed $\bar{\xi},\bar{\eta}$, and so we get equation (\ref{limitito}). 

Analogously, one can prove that

\be
\tr(\rho\prod_{j=1}^kW_{\bar{\xi}_j})=\lim_{N\to\infty}L\left(\prod_{j=1}^k\sum_{l=0}^N\frac{(i\bar{\xi}_j\cdot\sigma \bar{R})^l}{l!}\right),
\label{genara}
\ee

\noindent for any $k$.

Now, let $p=\prod_{k=1}^mp_k\in W_n$, with $p_k\in\{\frac{a_i+a_i^*}{2},\frac{a_i-a_i^*}{2i}\}_{i=1}^n$. Then, 

\be
\frac{L(p)}{i^m}=\frac{\partial^m}{\partial \epsilon_1,...,\epsilon_k} L(\prod_{k=1}^m\sum_{l=0}^\infty\frac{(i\epsilon_ks_k)^l}{l!})\mid_{\bar{\epsilon}=0}=\frac{\partial^m}{\partial \epsilon_1,...,\epsilon_k}\tr(\rho\prod_{k=1}^me^{i\epsilon_kp_k}))=\frac{\tr(\rho\pi(p))}{i^m}.
\ee

\noindent By taking linear combinations of the former expectation values, we thus have that $\tr(\rho \pi(h))=L(h)$ for any $h\in \W_n$.

It only rests to show that $\rho$ is a quantum state, i.e., it is trace-class and positive semidefinite.

\noindent From the quantum Bochner-Kinchin theorem \cite{Holevo}, we know that $\chi(\bar{\xi})$
is the characteristic function of a non normalized quantum state if and only if 

i) $\chi(\bar{\xi})$ is continuous at the origin.

ii) for any $r\in \N$, $\bar{\xi}_1,\bar{\xi}_2, ..., \bar{\xi}_r\in \R^{2n}$, and $c_1,c_2,...,c_r\in
\C$ the relation

\begin{equation}
\sum_{k,l=0}^rc_kc_l^*\chi(\bar{\xi}_k-\bar{\xi}_l)e^{i\bar{\xi}_k\sigma\bar{\xi}_l/2}\geq
0
\end{equation}

\noindent holds.

From (\ref{cosita}) we know that $\chi(\bar{\xi})$ is not only continuous
everywhere but even analytic. On the other hand, from (\ref{limitito}) we have that

\begin{equation}
\sum_{k,l=0}^rc_kc_l^*\chi(\bar{\xi}_k-\bar{\xi}_l)e^{i\bar{\xi}_k\sigma\bar{\xi}_l/2}=\lim_{m\to\infty}L(h_mh_m^*),
\end{equation}

\noindent where $h_m=\sum_kc_k\sum_{l=0}^m
\frac{(i\bar{\xi_k}\cdot\sigma \bar{R})^l}{l!}$. Since $h_m\in\W_n$, $L(h_mh_m^*)\geq 0$, and so the above limit is non-negative.

\end{proof}

\end{appendix}

\end{document}